\documentclass[conference]{IEEEtran}
\IEEEoverridecommandlockouts

\usepackage{cite}
\usepackage{amsmath,amssymb,amsfonts}
\usepackage{algorithmic}
\usepackage{graphicx}
\usepackage{amsthm}
\usepackage{textcomp}
\usepackage{xcolor}
\usepackage{breqn}
\def\BibTeX{{\rm B\kern-.05em{\sc i\kern-.025em b}\kern-.08em
    T\kern-.1667em\lower.7ex\hbox{E}\kern-.125emX}}

\newcommand{\bs}{\boldsymbol}
\newcommand{\bb}{\mathbb}

\newcommand{\cl}{\mathcal}
\newcommand{\ts}{\textstyle}

\newcommand{\tr}{\text{tr}}

\newcommand{\ie}{\emph{i.e.}, }

\DeclareMathOperator*{\sign}{sign}
\newcommand{\scp}[3][]{#1\langle #2, #3 #1\rangle}

\newtheorem{thm}{Theorem}
\newtheorem{proposition}[thm]{Proposition}

\begin{document}

\title{\vspace{-0mm}Random Features for Grassmannian Kernels\vspace{-0mm}
}



\author{\IEEEauthorblockN{Rémi Delogne and Laurent Jacques}
\IEEEauthorblockA{ISPGroup, INMA/ICTEAM, UCLouvain, Belgium\\
\{remi.delogne,laurent.jacques\}@uclouvain.be}\vspace{-5mm}
}


\maketitle


\begin{abstract}

The Grassmannian manifold $\cl G(k,n)$ serves as a fundamental tool in signal processing, computer vision, and machine learning, where problems often involve classifying, clustering, or comparing subspaces. In this work, we propose a sketching-based approach to approximate Grassmannian kernels using random projections. We introduce three variations of kernel approximation, including two that rely on \emph{binarised sketches}, offering substantial memory gains. We establish theoretical properties of our method in the special case of $\cl G(1,n)$ and extend it to general $\cl G(k,n)$. Experimental validation demonstrates that our sketched kernels closely match the performance of standard Grassmannian kernels while avoiding the need to compute or store the full kernel matrix. Our approach enables scalable Grassmannian-based methods for large-scale applications in machine learning and pattern recognition.

\end{abstract}

\begin{IEEEkeywords}
Grassmann, Kernel, Random Features
\end{IEEEkeywords}


\section{Introduction}
The Grassmannian manifold $\cl G(k, n)$, representing the space of all $k$-dimensional linear subspaces of $\bb{R}^n$, plays a fundamental role in many applications in signal processing, computer vision, and machine learning \cite{ji2015angular}\cite{Chang2012}\cite{harandi2014expanding}. Problems in these fields often involve classifying, clustering, or comparing subspaces, requiring the use of appropriate similarity measures or kernel functions. However, standard Euclidean methods fail due to the non-Euclidean nature of Grassmannians \cite{hamm2008GrassmanDiscriminant}, requiring the development of dedicated kernel-based approaches. 

Kernels on Grassmannians have been extensively studied, with notable examples including the projection kernel and the Binet-Cauchy kernel \cite{harandi2014expanding}. These kernels provide means to embed subspaces into a reproducing kernel Hilbert space (RKHS), allowing the application of powerful classification algorithms such as support vector machines (SVMs). However, the large storage cost required to store the kernel matrix often makes such methods impractical in large datasets. 

In this work, we propose a sketching mechanism that approximates subspace similarity by mapping projection matrices into a feature space through random embeddings. Inspired by the methodology of Rahimi and Recht \cite{Rahimi2007} also described in \cite{Tropp2015}, our approach provides an alternative way to apply Grassmannian kernels in large-scale settings by avoiding explicit computation of the Gram matrix which is beneficial when working with large datasets or when storing and computing full kernel matrices becomes infeasible.

Our main contributions are as follows. We propose three kernel approximation schemes for Grassmannian manifolds based on random projections, specifically using rank-one projections (ROP). Two of these schemes involve binarised embeddings, which drastically reduce both memory requirements and computational cost by replacing real-valued inner products with simple binary operations. We then analyse the expected behaviour of these approximations in the special case of $\cl G(1,n)$, showing closed-form dependence on principal angles, and we extend key insights to the general $\cl G(k,n)$ setting. Finally we provide experimental validation on synthetic and real-world datasets, demonstrating that even fully binarised sketches can match the classification accuracy of exact Grassmannian kernels for sufficiently large sketches.


The remainder of the paper is structured as follows: In Sec.~II, we provide background on Grassmannian manifolds and existing kernel methods. Sec.~III introduces our sketching-based  approximation, detailing its mathematical formulation and properties. Sec.~IV presents experimental results, and Sec.~V discusses applications and future work.

\section{Background on Grassmannians and kernels}

\paragraph{Grassmannian manifolds}
The Grassmannian manifold $\cl G(k, n)$ is the set of all $k$-dimensional linear subspaces of $\bb{R}^n$. These manifolds appear naturally in various fields like subspace clustering, dimensionality reduction, and signal processing, where data can be represented by subspaces instead of vectors.

An element of $\cl G(k, n)$ is commonly represented by an $n \times k$ matrix $\bs U$ with orthonormal columns, satisfying $\bs U^\top \bs U = \bs I_k$, where $\bs I_k$ is the $k \times k$ identity matrix. However, since the choice of basis is not unique, an equivalent representation is given by the projection matrix: $\bs P = \bs U \bs U^\top$, where $\bs P$ is an $n \times n$ symmetric and idempotent matrix, satisfying $\bs P^2 = \bs P$. This formulation ensures that comparisons between elements of $\cl G(k, n)$ are invariant to basis choices.

The Grassmannian manifold is a \emph{Riemannian manifold} with a well-defined geometric structure, making it a natural setting for tasks requiring subspace comparisons. It has applications in areas such as image recognition, motion segmentation, and dictionary learning. 

\paragraph{Principal angles and distance metrics}
To measure the similarity between two subspaces $\bs P, \bs Q \in \cl G(k, n)$, we use \emph{principal angles} $\theta_1, \dots, \theta_k$, which quantify the degree of alignment between the subspaces. These angles are obtained from the singular value decomposition (SVD), $\bs P^\top \bs Q = \bs V_1 \bs{\Sigma} \bs V_2^\top,$ where $\bs{\Sigma} = \text{diag}(\cos \theta_1, \dots, \cos \theta_k)$ contains the cosines of the principal angles.

Several distance metrics on $\cl G(k, n)$ are derived from these angles, for example the \emph geodesic distance,
    \begin{equation}
        \ts d_g(\bs P, \bs Q) = \big(\sum_{i=1}^{k} \theta_i^2\big)^{\frac{1}{2}},
    \end{equation}
    which measures the shortest path length between the subspaces on the manifold.

Such distance metrics have various applications, the key point is that they all depend on the principal angles between subspaces. This dependence is crucial for defining similarity measures and kernel functions on Grassmannians.

\paragraph{Kernels on Grassmannians}
Kernel methods provide a powerful approach for analysing data in structured spaces such as Grassmannians. By embedding the subspaces into a reproducing kernel Hilbert space (RKHS), these methods enable the application of well-established machine learning algorithms \cite{Gruber2006GrassmanClustering}.

Several kernels have been proposed for $\cl G(k, n)$, most interestingly for us is the \emph{projection kernel} defined as

\begin{equation}\label{eq;projectionkernel}
    \kappa(\bs P, \bs Q) = (\det(\bs P^\top \bs Q))^2.
\end{equation}

These kernels have been successfully applied in areas such as object recognition, activity classification, and subspace-based clustering. However as mentioned earlier, the large storage and computational costs associated with the storage, computation and handling of the kernel matrix often imply poor scalability of such methods. This is where random features approximations such as in \cite{Rahimi2007} can be useful.

\paragraph{Sketching for kernel approximation}
We propose a random projection-based sketching mechanism to approximate Grassmannian kernels. The idea is to construct an embedding of projection matrices that preserves principal angle-based similarity while avoiding direct computation of the Gram matrix. By mapping each projection matrix into a compact representation, we enable efficient kernel-like computations without storing the full $N \times N$ Gram matrix.

We rely on rank-one projections \cite{cai2015ROP}\cite{delogne2023signal} and study how they can approximate similarity functions (or kernels) between spaces \cite{delogne2024quadratic}. To further reduce computational complexity, we introduce a binarised sketch that retains only the sign information of the projection. This approach enables highly efficient storage and computation, making it particularly attractive for large-scale applications.

In the next section, we begin by analysing the special case of $\cl G(1, n)$, where subspaces correspond to 1-dimensional lines in $\bb{R}^n$. This setting provides a fundamental insight into the behaviour of sketched kernels before extending to higher-dimensional subspaces.

\section{Sketching-Based Kernel Approximations}

\subsection{Random projections for Grassmannian kernels}
Random projections have been widely used as an efficient method to simplify calculations in kernel methods, particularly in the context of large-scale learning problems. Inspired by the original work of Rahimi and Recht \cite{Rahimi2007}, we use random projection techniques to the setting of Grassmannian manifolds. 

Given a subspace $\bs P \in \cl G(k, n)$, our goal is to construct a sketch that preserves principal angle-based similarity. Instead of working directly with projection matrices, we define a structured embedding using rank-one random projections.

More formally, let $\bs a_i, \bs b_i \in \bb{R}^n$ be independent standard Gaussian vectors for $i=1,\dots,m$. We define the random projection operator $\cl R:\cl G(k,n)\to\bb R^m$ as
\begin{equation}
    \ts \cl R(\bs P) = \big( \scp{\bs a_i\bs b_i^\top}{\bs P } \big)_{i=1}^{m}.
\end{equation}
Each entry of the sketched representation $\cl R(\bs P)$ corresponds to an inner product between the rank one matrix $\bs a_i\bs b_i^\top$ and the subspace, effectively capturing information about the structure of $\bs P $. As we shall see later, the parameter $m$ can be used to control the quality of the random projection. As $m$ increases, the approximation error of the embedding will decrease.

\subsection{Three kernel approximation mechanisms}
We introduce three variations of the kernel approximation mechanism, each offering different trade-offs between approximation accuracy and computational efficiency. We denote kernel functions using $\kappa_i$ and their approximations using $\tilde\kappa_i$.

\subsubsection{Non-binarised approximation}
The first and most straightforward approach is to retain the real-valued sketched vector $\cl R(\bs P)$. The resulting kernel approximation is given by:
\begin{equation}
    \tilde\kappa_1(\bs P, \bs Q) = \scp{\cl R(\bs P)}{\cl R(\bs Q)}.
\end{equation}
This method approximates the projection kernel \eqref{eq;projectionkernel}. It has the highest fidelity to the original subspace similarity but still requires storage of continuous values.

\subsubsection{Semi-binarised approximation}

This approach is inspired by the asymmetric sketching mechanism proposed in \cite{Schellekens2021}, where an asymmetric sketching mechanism is used for data storage and incoming queries. Specifically, we sketch the dataset using a binarised embedding $\cl S:\cl G(k,n)\to\bb R^m : \bs P\mapsto \text{sign}(\cl R(\bs P))$,
where the sign operator is applied componentwise to the sketched vector $\cl R(\bs P)$, and compute inner products with non-binarised embeddings of query points, \ie we define the kernel
\begin{equation}
    \tilde\kappa_2(\bs P, \bs Q) = \scp{\cl S(\bs P)}{\cl R(\bs Q)}.
\end{equation}
Although $\tilde\kappa_2$ is not a valid kernel in the strict sense (\ie it is not symmetric\footnote{Note that one can still use a symmetric kernel based on $\cl R$ and $\cl S$ by considering $\tilde\kappa_2'(\bs x,\bs y):=\big(\scp{\cl R(\bs x)}{\cl S(\bs y)}+\scp{\cl S(\bs x)}{\cl R(\bs y)}\big)/2$.}), it is a kernel \emph{in expectation}, meaning it converges to a well-defined limit as $m \to \infty$.

This semi-binarised design is particularly attractive in asymmetrical scenarios such as low-bandwidth settings, where sensors transmit compact binary sketches to a server that retains higher precision data (described in details in \cite{Schellekens2021}). It enables efficient scalar product computation while reducing communication and storage costs on the sensor or server side. 



\subsubsection{Fully binarised approximation}
In the most extreme case, we enforce full binarisation of the sketch with corresponding kernel approximation given by:
\begin{equation}
    \tilde\kappa_3(\bs P, \bs Q) = \scp{\cl S(\bs P)}{\cl S(\bs Q)}.
\end{equation}
While this approach significantly reduces memory the resulting kernel (obtained in expectation) is not easy to characterise as we will explain.

\subsection{Properties of the Sketched Kernels}
Each of the proposed kernels exhibits a fundamental dependence on the principal angles between subspaces. For the case of $\cl G(1, n)$, we can derive closed-form expressions for the expected kernel values in terms of the single principal angle.

\emph{Positive-definiteness and symmetry}: Standard kernel functions are required to be positive definite and symmetric \cite{Scholkopf2002}. The first two kernel functions, $k_1$ and $k_2$ are provably positive-definite because, as we will see, they can both be expressed by a scalar product. For $k_3$ it is less immediate, but experimentally it seems that it retains this property when $k=1$.

\begin{proposition}{(Dependence on the principal angle in 1 dimension)}
For $\bs P,\bs Q\in \cl G(1, n)$, where $ P_1 = \bs x \bs x^\top$ and $\bs Q = \bs y \bs y^\top$, with $\|\bs x\|=\|\bs y\|=1$ at an angle $\theta$, we have:
\begin{equation}\label{eq:k11d}
    \mathbb{E}[\tilde\kappa_1(\bs P,\bs Q)] =\scp{\bs P}{\bs Q}_F= \cos^2(\theta):=\kappa_1(\bs P,\bs Q),
\end{equation}
\begin{equation}\label{eq:k21d}
    \ts \mathbb{E}[\tilde\kappa_2(\bs P,\bs Q)]=\frac2\pi\cos^2(\theta):=\kappa_2(\bs P,\bs Q), and\end{equation}
\begin{equation}\label{eq:k31d}
    \ts \mathbb{E}[\tilde\kappa_3(\bs P, \bs Q)] = \left(\frac{2}{\pi} \theta - 1 \right)^2:=\kappa_3(\bs P,\bs Q).
\end{equation}
\end{proposition}

\begin{proof}
    We will prove \eqref{eq:k11d} and \eqref{eq:k21d} in the next section as they are a trivial consequence of the $n-$dimensional case. As for \eqref{eq:k31d}, let $m=1$\footnote{Since we are taking expectation we can do this WLOG.}, we have, 
    \begin{multline}
        \label{eq:expk3}
        \bb E[\scp{\cl S(\bs x)}{\cl S(\bs y)}] =\\
        \bb E[\sign(\bs a^\top\bs x)\sign(\bs a^\top\bs y)\sign(\bs b^\top\bs x)\sign(\bs b^\top\bs y)]
    \end{multline} 
    Since $\bs a^\top\bs x$, $\bs a^\top\bs y$ are independent from $\bs b^\top\bs x$, $\bs b^\top\bs y$ we split the product of expectations. Furthermore we know from lemma 3.2 in \cite{goemans1995ImprovedApproximation} that $\bb P[\sign(\bs a^\top\bs x)\neq\sign(\bs a^\top\bs y)]=1-\frac\theta\pi$, with $\theta$ the angle between $\bs x$ and $\bs y$. Thus $\bb E[\sign(\bs a^\top\bs x)\sign(\bs a^\top\bs y)]=\frac{2\theta}{\pi}-1$ and $\eqref{eq:expk3}$ becomes~$(\frac{2\theta}{\pi}-1)^2$.
\end{proof}

\subsection{Extension to Higher-Dimensional Subspaces}

\begin{proposition}
    In $\cl G(k, n)$, the kernel approximations retain their dependence on principal angles. For the first case we have:
\begin{equation}\label{eq:k1nd}
    \bb E[\tilde\kappa_1(\bs P, \bs Q)] = \ts \sum_{i=1}^{k} \cos^2(\theta_i) 
\end{equation}
For the semi-binarised case, we have 
\begin{equation}\label{eq:k2nd}
    \bb E[\tilde\kappa_2(\bs P, \bs Q)] = \ts c_k \sqrt{\frac{2}{\pi}} \scp{\bs P}{\bs Q},
\end{equation}
where $c_k=\bb E\|\bs g\|$ with $\bs g\sim\cl N(0,\bs I_k)\in\bb R^k$.
\end{proposition}
\begin{proof}
    Again let us take $m=1$ WLOG. The proof of \eqref{eq:k2nd} is straightforward: \[\bb E[\scp{\cl R(\bs P)}{\cl R(\bs Q)}]=\bb E[\tr(\bs a^\top\bs P\bs b\bs b^\top\bs Q\bs a)],\] and since the trace is invariant under cyclic permutations and $\bb E[\bs a\bs a^\top]=\bb E[\bs b\bs b^\top]=\bs I$, we are left with \[\ts \bb E[\scp{\cl R(\bs P)}{\cl R(\bs Q}]=\scp{\bs P}{\bs Q}_F=\sum_{i=1}^k\cos^2(\theta_i).\]
For \eqref{eq:k2nd}, we have
$$
\ts \mathbb{E}[\scp{\cl S(\bs P)}{\cl R(\bs Q)}] = \mathbb{E}[\sign(\bs a^\top \bs P \bs b) \bs b^\top \bs Q^\top \bs a].
$$
Since $\bs P$ is a projection matrix, we decompose it as $\bs P = \bs U \bs U^\top$, where $\bs U \in \mathbb{R}^{n \times k}$, allowing us to rewrite:
\begin{equation*}
\mathbb{E}[\sign(\bs a^\top \bs U \bs U^\top \bs b) \bs b^\top \bs Q^\top \bs a] := \mathbb{E}[\sign(\tilde{\bs a}^\top \tilde{\bs b}) \bs b^\top \bs Q^\top \bs a].
\end{equation*}
Let $\bs U^\perp$ denote the orthogonal complement of $\bs U$, such that $\bs a = \bs U \tilde{\bs a} + \bs U^\perp \tilde{\bs a}^\perp$ and $\bs b = \bs U \tilde{\bs b} + \bs U^\perp \tilde{\bs b}^\perp$. Substituting these into the inner product expansion:
\begin{equation*}
\bs b^\top \bs Q^\top \bs a = (\bs U \tilde{\bs a} + \bs U^\perp \tilde{\bs a}^\perp)^\top \bs Q (\bs U \tilde{\bs b} + \bs U^\perp \tilde{\bs b}^\perp),
\end{equation*}
and taking expectations, we find that terms involving cross-products with orthogonal components vanish, leaving:
\begin{equation*}
\mathbb{E}[\sign(\bs a^\top \bs P \bs b) \bs b^\top \bs Q^\top \bs a] = \mathbb{E}[\sign(\tilde{\bs a}^\top \tilde{\bs b}) \tilde{\bs b}^\top \bs U^\top \bs Q \bs U \tilde{\bs a}].
\end{equation*}
Let $\bs U^\top \bs Q \bs U = \bs A$ and $s=\sign(\tilde{\bs a}^\top \tilde{\bs b})$, we rewrite the expectation using the trace property:
\begin{equation*}
\tr(\mathbb{E}[s\tilde{\bs a}\tilde{\bs b}^\top \bs A]) = \tr(\mathbb{E}[s\tilde{\bs a}\tilde{\bs b}^\top] \bs A).
\end{equation*}
Since $\tilde{\bs a}$ and $\tilde{\bs b}$ are independent Gaussian vectors, the off-diagonal element of $\mathbb{E}[s \tilde{\bs a} \tilde{\bs b}^\top]$ must vanish and we get, for some $\alpha$, $\mathbb{E}[s \tilde{\bs a} \tilde{\bs b}^\top] = \alpha \bs I_k$ by rotational invariance. Thus, $\tr(\mathbb{E}[s \tilde{\bs a} \tilde{\bs b}^\top]) = \alpha k = \mathbb{E}[|\tilde{\bs a}^\top \tilde{\bs b}|]$.
For $k=1$, we obtain $\mathbb{E}[|\tilde{\bs a}^\top \tilde{\bs b}|] = \frac{2}{\pi}$,
thus proving \eqref{eq:k21d}. As for general $k$, since $\tilde{\bs b} = \bs U \bs U^\top \bs b$ is a Gaussian random vector with unit variance components in the space spanned by $\bs U$, we get by the law of total expectation $\mathbb{E}[|\tilde{\bs a}^\top \tilde{\bs b}|] = \mathbb{E}_{\tilde{\bs a}}(\mathbb{E}_{\tilde{\bs b}}[|\tilde{\bs a}^\top \tilde{\bs b}|]) = \mathbb{E}_{\tilde{\bs a}} \sqrt{\frac{2}{\pi}} \|\tilde{\bs a}\| = \sqrt{\frac{2}{\pi}} c_k$, with $c_k$ defined above.
Finally, using the cyclic property of the trace,
$\tr(\bs U^\top \bs Q \bs U) = \tr(\bs Q \bs U \bs U^\top) = \tr(\bs Q \bs P) = \scp{\bs P}{\bs Q}$,
proving that $\mathbb{E}[\scp{\cl S(\bs P)}{\cl R(\bs Q)}] = \sqrt{\frac{2}{\pi}} \mathbb{E}[|\bs g|] \scp{\bs P}{\bs Q}$.
\end{proof}

For the fully binary case, there seems to be no closed-form expression for the expected value, we can however state:

\begin{proposition}
    For $\bs P,\bs Q\in\cl G(k,n)$, with principal angles $\theta_1,\ldots,\theta_k$ between them, $\kappa_3(\bs P,\bs Q)=f(\theta_1,\ldots,\theta_k)$ for some function $f$. In other words $\kappa_3(\bs P,\bs Q)$ only depends on the principal angles between $\bs P$ and $\bs Q$.
\end{proposition}
\begin{proof}
    By exploiting the rotational invariance of the Gaussian distribution of both $\bs a$ and $\bs b$, we know that, given an arbitrary orthogonal matrix $\bs R \in O(n)$, $\sign(\bs a^\top\bs P\bs b)\sign(\bs a^\top\bs Q\bs b)$ and $\sign(\bs a^\top\bs R^\top\bs P\bs R\bs b)\sign(\bs a^\top\bs R^\top\bs Q\bs R\bs b)$ share the same distribution. Thus in expectation, $\kappa_3(\bs P,\bs Q)=\mathbb{E}[\scp{\cl S(\bs P)}{\cl S(\bs Q)}] = \kappa_3(\bs R^\top\bs P\bs R,\bs R^\top\bs Q\bs R)$. By \cite{conway1996packing}, we can pick $\bs R$ such that $\bs R^\top\bs P\bs R = \big( \bs I_k, \textbf{0}_{n-k})^\top \big( \bs I_k, \textbf{0}_{n-k})$ and ${\bs R^\top\bs Q\bs R = \big( \bs C, \bs S, \textbf{0}_{n-2k})^\top \big( \bs C, \bs S, \textbf{0}_{n-2k})}$, with $\bs C:={\rm diag}(\cos \theta_1,\ldots,\cos \theta_k)$ and $\bs S:={\rm diag}(\sin \theta_1,\ldots,\sin \theta_k)$, showing that $\kappa_3(\bs R^\top\bs P\bs R,\bs R^\top\bs Q\bs R)$, and thus $\kappa_3(\bs P,\bs Q)$, only depends on the principal angles.
\end{proof}



\begin{figure*}[!t]
    \centering
    \includegraphics[width=0.45\textwidth]{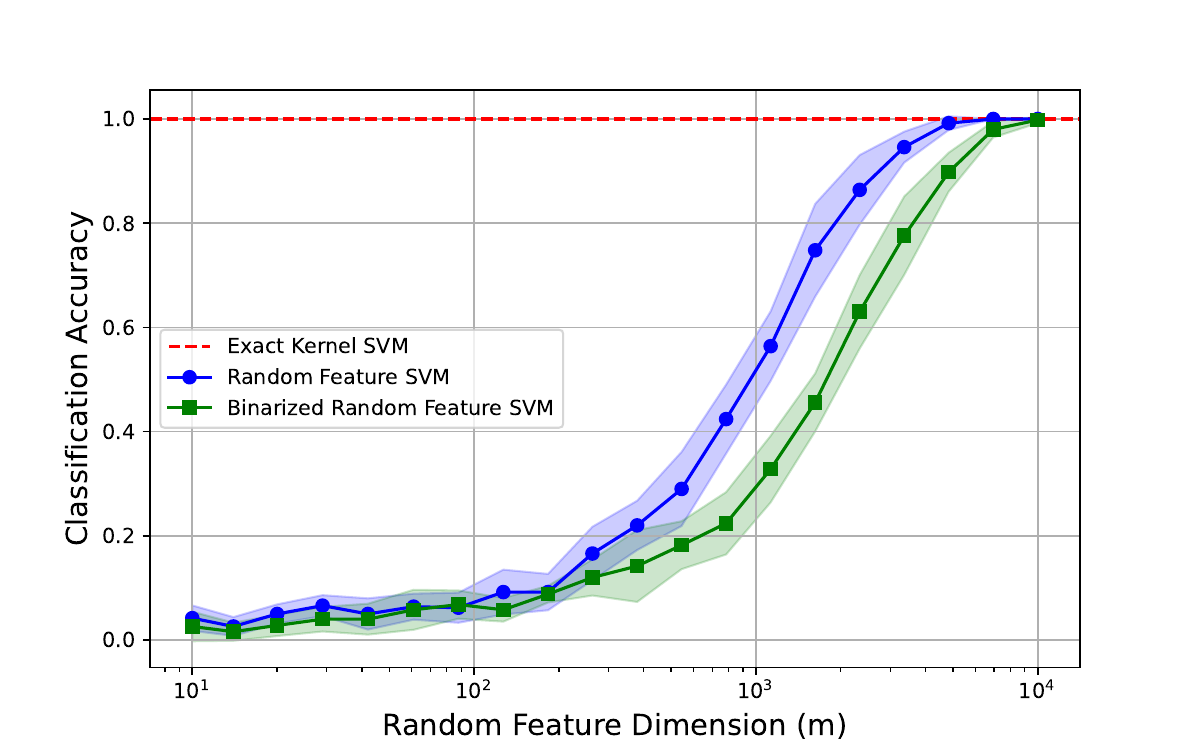}
    \includegraphics[width=0.45\textwidth]{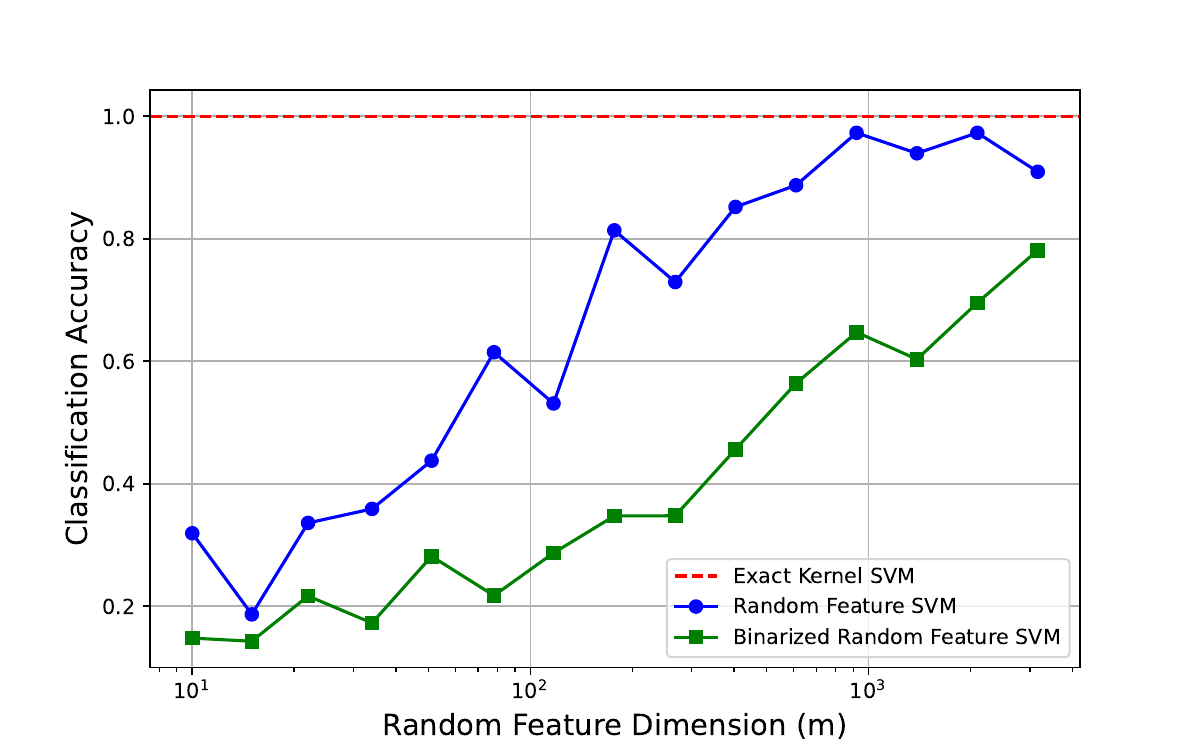}
    \caption{(left) Classification accuracy on synthetic subspaces as a function of $m$. (right) Classification accuracy on ETH-80 dataset.}
    \label{fig:classif}
\end{figure*}

\subsection{Benefits of approximation}

As observed in \cite{Rahimi2007}, \cite{Lopezpaz2014} and \cite{Tropp2015}, the primary advantage of approximating kernel functions via random features is not necessarily a reduction in computational complexity or dataset size, but rather circumventing the storage of a Gram matrix which becomes infeasible when the number of instances to classify is large. Random projections offer a viable alternative that may allow Grassmannian kernels to be applied more broadly in large-scale applications.

\section{Experimental Validation}

We evaluate the proposed kernel approximations on both synthetic and real-world data. Our primary classification method is support vector machines (SVMs), trained either with a precomputed Gram matrix or directly in the sketched feature space, eliminating the need for computing and storing the Gram matrix. We also validate the approximation quality of $\tilde{\kappa}_2$ using a nearest subspace classifier.

\subsection{Experiments on Artificial Subspaces}

Synthetic subspaces of dimension $9$ in $\bb{R}^{1024}$ $\big($in $\cl G(9,1024)$$\big)$ are generated by sampling random vectors and orthonormalising\footnote{Orthonormalisation incurs an additional computational cost but is in line with existing methods such as \cite{ChangFaceRecognition2007}.} them with QR decomposition. Test subspaces are obtained by perturbing the training ones with additive noise.

We train SVM classifiers in two settings: (i) using the exact kernel $\kappa_1$ with a precomputed Gram matrix, and (ii) using its approximations $\tilde{\kappa}_1$ and $\tilde{\kappa}_3$ without Gram computation, operating directly in the sketched feature space. This demonstrates the possibility of removing the Gram matrix in kernel methods on Grassmannian. Despite lacking a theoretical expectation, $\tilde{\kappa}_3$ also achieves strong classification performance. 

To assess the quality of $\tilde{\kappa}_2$, we use a nearest subspace classifier, which assigns test subspaces to the class with the highest similarity. For $m\geq 5000$, the classification accuracy using $\tilde{\kappa}_2$ is indistinguishable from that obtained with $\kappa_2$, confirming its validity as an approximation. 

To account for randomness in the sketching process, we average results over 10 independent trials, and the standard deviations are reported in Fig.~\ref{fig:classif}(left). For sufficiently large $m$, $\tilde{\kappa}_1$ and $\tilde{\kappa}_3$ match the performance of $\kappa_1$, with $\tilde{\kappa}_3$ requiring more features but offering storage benefits due to its binarisation.

\subsection{Experiments on ETH-80 Dataset}

We apply the same methodology to the ETH-80 dataset \cite{leibe2003ETH80dataset}\cite{chen2020covariance}, which consists of images from 8 object categories (apples, pears, cows, \ldots) captured from multiple angles. Following \cite{ji2015angular} where the authors explain that objects pictured from different angles can be modelled by a 9-dimensional subspace of $\bb R^n$, each object class is modelled as a $9$-dimensional subspace in $\bb R^n$. Training subspaces are extracted from greyscale images resized to $n=32 \times 32$, and classification is performed on subspaces constructed from unseen images. As in the synthetic case, we compare SVM classifiers trained with $\kappa_1$, $\tilde{\kappa}_1$, and $\tilde{\kappa}_3$, while using nearest subspace classification to validate $\tilde{\kappa}_2$. Fig.~\ref{fig:classif}(right) shows that both $\tilde{\kappa}_1$ and $\tilde{\kappa}_3$ perform well both converging to $\kappa_1$'s accuracy- at large $m$. The nearest subspace classifier confirms that for $m \geq 10^3$, $\tilde{\kappa}_2$ achieves the same classification accuracy as $\kappa_2$. 

One should keep in mind that through the binarisation process we reduce the storage by a factor of 64 compared to storing in double-precision as is the case in most computers, making our approach more efficient even when $m$ seems large.


\section{Conclusion and Future Work}

We have shown that Grassmannian kernels can be approximated using random features, enabling classification without computing or storing a Gram matrix. Through experiments on synthetic and real-world data, we demonstrated that the approximations $\tilde{\kappa}_1$ and $\tilde{\kappa}_3$ achieve similar performances to $\kappa_1$ when $m$ is large, with $\tilde{\kappa}_3$ offering additional storage benefits due to its binarisation. We also validated the approximation quality of $\tilde{\kappa}_2$ using another classifier, showing that it can reproduce the classification results of $\kappa_2$ for large $m$. While $\tilde{\kappa}_2$ doesn't provide computational advantages in SVM classification, it confirms the validity of the sketching process.

Future work directions include the theoretical analysis of $\tilde{\kappa}_3$ beyond the $k=1$ case, as its strong empirical performance suggests potential for further study. Another direction is developing direct subspace construction from sketches, to allow classification to be performed without explicitly forming projection matrices, further reducing computational power.

\bibliographystyle{plain}
\bibliography{BigBibliography}

\end{document}